\documentclass[letterpaper, 10 pt, conference]{ieeeconf}  

\IEEEoverridecommandlockouts                              
\title{\LARGE \bf 
Towards Resilience for Multi-Agent $QD$-Learning
}

\author{Yijing Xie, Shaoshuai Mou, and Shreyas Sundaram*
\thanks{Y. Xie is a Lillian Gilbreth Postdoctoral Fellow at the College of Engineering, S. Mou is with the School of Aeronautics and Astronautics and S. Sundaram is with the School of Electrical and Computer Engineering, Purdue University, West Lafayette, IN 47907 USA (e-mail: xie382@purdue.edu, mous@purdue.edu, sundara2@purdue.edu).}
}

\usepackage{amsfonts}
\usepackage{xcolor}
\usepackage{amsmath}
\usepackage{comment}
\usepackage{graphicx}
\usepackage{amssymb}
\usepackage{cite}
\usepackage{algpseudocode}
\usepackage{algorithm}

\DeclareMathOperator*{\argmin}{arg\,min}

\newtheorem{defn}{Definition}
\newtheorem{thm}{Theorem}
\newtheorem{lem}{Lemma}
\newtheorem{prop}{Proposition}

\newtheorem{assum}{Assumption}
\newtheorem{remark}{Remark}

\def\T{^{\top}}
\def\ie{{\rm i.e.}}

\def\as{{\rm a.s.}}
\def\nn{\nonumber}

\begin{document}

\maketitle
\thispagestyle{plain}
\pagestyle{plain}

\begin{abstract}
 This paper considers the multi-agent reinforcement learning (MARL) problem for a networked (peer-to-peer) system in the presence of Byzantine agents.  We build on an existing distributed $Q$-learning algorithm, and allow certain agents in the network to behave in an arbitrary and adversarial manner (as captured by the Byzantine attack model). Under the proposed algorithm, if the network topology is $(2F+1)$-robust and up to $F$ Byzantine agents exist in the neighborhood of each regular agent, we establish the almost sure convergence of all regular agents' value functions to the neighborhood of the optimal value function of all regular agents. For each state, if the optimal $Q$-values of all regular agents corresponding to different actions are sufficiently separated, our approach allows each regular agent to learn the optimal policy for all regular agents.
\end{abstract}

\section{Introduction}

In multi-agent reinforcement learning (MARL), multiple agents observe the outcome of interactions with an environment, and use those observations to learn optimal control policies to achieve long-term goals. By working cooperatively, agents are able to optimize a common long-term reward which is an aggregate of all agents' private rewards \cite{kar2013cal,zhang2018fully,qu2020scalable,lin2020distributed}. The authors of \cite{kar2013cal} approach the MARL problem by a distributed $Q$-learning algorithm, in which each agent maintains a $Q$-value estimate for every state-action pair. The convergence of the $Q$-value estimates to the optimal $Q$ values is guaranteed. Subsequently, \cite{zhang2018fully} proposes actor-critic algorithms with convergence guarantees using linear functions to parameterize $Q$-value estimates. Each agent shares its parameter instead of $Q$-value estimates to its neighbors.
By exploiting the network structure, \cite{qu2020scalable} proposes a scalable actor-critic algorithm where each agent maintains $Q$-value estimates only for state-action pairs within its multi-hop neighbors. This result has been further extended in \cite{lin2020distributed} to the case of time-varying networks.

Algorithms for multi-agent systems are typically robust against benign failures of individual agents as long as the underlying network is connected. However, the dependence of these algorithms on local coordination among neighbors also raises a major security concern that the presence of one or more malicious agents under cyberattacks could compromise the entire algorithm \cite{sundaram2018distributed}. It is thus imperative to develop algorithms that are \emph{resilient}, which refers to algorithms' ability to withstand the compromise of a subset of the agents and still ensure some notion of correctness \cite{leblanc2013resilient}. Resilient algorithms against various types of attackers for networked systems have been proposed for different problems such as consensus \cite{pasqualetti2012consensus,leblanc2013resilient,wang2019resilient}, distributed optimization \cite{sundaram2018distributed,zhao2019resilient,kuwaranancharoen2020byzantine} and distributed learning \cite{chen2017distributed,blanchard2017byzantine,yin2018byzantine,yang2019byrdie,yang2019bridge}. Within the class of resilient distributed learning algorithms, some papers assume a client-server architecture where a central agent collects information from all other agents and broadcasts new information back to other agents\cite{chen2017distributed,blanchard2017byzantine,yin2018byzantine}.
Other algorithms such as ByRDiE in \cite{yang2019byrdie} and BRIDGE in \cite{yang2019bridge} are designed based on the peer-to-peer (P2P) architecture, where there is no central agent to coordinate all other agents, and all agents exchange information with neighbors. Very recently, resilient algorithms for MARL in the presence of Byzantine agents are proposed in \cite{lin2020toward} and \cite{wu2020byzantine}. Specifically, \cite{lin2020toward} considers the fully cooperative MARL problem for a networked system in the client-server architecture with a reliable central agent. The paper \cite{wu2020byzantine} considers the policy evaluation problem in the P2P architecture. By assuming a bounded reward variation between the local reward of each agent and the global averaged reward of all agents, they obtain a learning error, which is related to the bound of the reward variation, network structure and discounting factor.

In this paper, we propose a resilient $QD$-learning algorithm for a networked system in the presence of Byzantine agents. The main motivation is that the $QD$-learning algorithm generally fails even in the presence of a single adversarial agent. We first extend the distributed $Q$-learning algorithm for undirected networks \cite{kar2013cal} to time-varying directed networks. We then build on that to create a resilient $QD$-learning that is capable of tolerating Byzantine attacks. For each regular agent, we establish the almost sure convergence of the value function to the neighborhood of the optimal value function of all regular agents under certain conditions on the graph topology. For each state, we show that if the optimal $Q$-values corresponding to different actions are sufficiently separated, each regular agent can learn the optimal policy for all regular agents.

\textit{Organization:}  The fully cooperative MARL problem for a networked system is formulated in Section \ref{sec: problem formulation}. The extension of the $QD$-learning algorithm to a time-varying directed network is presented in Section \ref{sec: directed QD}. In section \ref{sec: resilient QD}, we first characterize limitations on the performance of the $QD$-learning algorithm in the presence of adversaries. Then we introduce a new resilient $QD$-learning algorithm and provide the main result. We conclude the paper in Section \ref{sec: conclusion}.

\textit{Notation:} Let $\mathbb{N}$ denote the set of all natural numbers and $\mathbb{R}$ the set of all real values. Let $\mathbb{R}^k$ denote the $k$-dimensional Euclidean space. Throughout, the probability space $(\Omega,\mathcal{F})$ supports all random objects. For a collection $\mathcal{J}$ of random objects, $\sigma(\mathcal{J})$ is the smallest $\sigma$-algebra with respect to which all the random objects in $\mathcal{J}$ are measurable. Probability and expectation on $(\Omega,\mathcal{F})$ are denoted by $\mathbb{P}(\cdot)$ and $\mathbb{E}(\cdot)$, respectively. All inequalities involving random objects are interpreted almost surely (\as).

\section{Problem Formulation}
\label{sec: problem formulation}

Consider a networked system consisting of $N$ agents, in which each agent can only communicate with certain other agents called \emph{neighbors}. The inter-agent communication network is represented by a time-invariant graph $\mathcal{G}=(\mathcal{V},\mathcal{E})$. Here $\mathcal{V}=\{v_1,v_2,\cdots,v_N\}$ denotes the node set with each node representing an agent; $\mathcal{E}\subset\mathcal{V}\times\mathcal{V}$ denotes a set of edges corresponding to the neighbor relations. A graph is said to be undirected if $(v_n,v_l)\in\mathcal{E}\Leftrightarrow(v_l,v_n)\in\mathcal{E}$, and directed otherwise. The neighbor set of $v_n$ is denoted by $\mathcal{N}_n=\{v_l\in\mathcal{V}|(v_l,v_n)\in\mathcal{E}\}$. 

Let $\{\mathbf{x}_t\}$ be a controlled Markov chain taking values in a finite state space $\mathcal{X}=\{1,2,\cdots,M\}$, and $\mathcal{U}$ be the finite set of control actions. The state transition is governed by \[\mathbb{P}(\mathbf{x}_{t+1}=j|\mathbf{x}_t=i,\mathbf{u}_t=u)=p_{ij}^u, \;\forall i,j\in\mathcal{X}, u\in\mathcal{U},\] where $\sum_{j\in\mathcal{X}}p_{ij}^u=1$ for all $i\in\mathcal{X}$. The private information $c_{n}(i,u)$ is the random one-stage cost of agent $v_n$ when control $u$ is applied at state $i$. A stationary control policy $\pi$ is a mapping from $\mathcal{X}$ to $\mathcal{U}$, where $\{\mathbf{u}_t\}$ satisfies $\mathbf{u}_t=\pi(\mathbf{x}_t)$. For a stationary policy $\pi$, the state process $\{\mathbf{x}_t^{\pi}\}$ evolves as a homogeneous Markov chain with $\mathbb{P}(\mathbf{x}_{t+1}^\pi=j|\mathbf{x}_t^\pi=i)=p_{ij}^{\pi(i)}$. For a stationary policy $\pi$ and initial state $i$ of the process $\{\mathbf{x}_t^\pi\}$, the infinite horizon discounted cost of agent $v_n$ is 
\[ V_{i,\pi}^n=\limsup\limits_{T\to\infty}\mathbb{E}\bigg[\sum_{t=0}^T\gamma^tc_n(\mathbf{x}_t^\pi,\pi(\mathbf{x}_t^\pi))|\mathbf{x}_0^\pi=i\bigg],  \]
where $\gamma\in(0,1)$ is the discounting factor. 

In the situation where all agents are reliable (\ie, non-adversarial), the $QD$-learning algorithm in \cite{kar2013cal} ensures that each agent eventually learns the optimal value function of all agents $\mathbf{V}^*=[V_1^*\;V_2^*\;\cdots\;V_M^*]\T$ and the associated optimal policy $\pi^*$ with 
\[V_i^*=\inf\limits_\pi \frac{1}{N}\sum_{v_n\in\mathcal{V}} V_{i,\pi}^n,\;\forall i\in\mathcal{X}.\]

In this paper, we consider the problem in the presence of adversarial agents. The node set $\mathcal{V}$ is partitioned into a set of regular nodes $\mathcal{R}$ and a set of adversarial nodes $\mathcal{A}=\mathcal{V}\setminus\mathcal{R}$ which is unknown a priori to the regular nodes.
It is generally impossible to learn $\mathbf{V}^*$ in the presence of  adversaries (as we will show later), since their local costs can never be accurately inferred. Instead, we will design a resilient $QD$-learning algorithm to approximately learn the optimal value function of all regular agents
$\mathbf{V}^{\mathcal{R}*}=[V_1^{\mathcal{R}*}\;V_2^{\mathcal{R}*}\;\cdots\;V_M^{\mathcal{R}*}]\T$
and the associated optimal policy $\pi^{\mathcal{R}*}$ with
\[V_i^{\mathcal{R}*}=\inf_{\pi}\frac{1}{|\mathcal{R}|}\sum_{v_n\in\mathcal{R}}V_{i,\pi}^n,\;\forall i\in\mathcal{X}.\]

\begin{remark}
A significant challenge in MARL settings where the agents themselves apply inputs is that the inputs applied by each agent will affect the state, but may not be visible to other agents. To deal with this, the majority of existing work assumes either that the inputs applied by all agents are globally visible \cite{kar2013cal,zhang2018fully,lin2020toward}, that there is a global controller \cite{kar2013cal}, or that there are no inputs at all \cite{wu2020byzantine}, with limited exceptions \cite{zhang2018fully,lin2020toward}.  In settings where agents may be adversarial (as in our work), the issue of agents applying inputs themselves incurs additional complexity, as the adversarial agents' inputs can no longer be easily predicted.  In this paper, we thus make the assumption of a global controller (whose actions are visible to all agents) in order to focus on the issue of resiliently learning the optimal policy; as we will see, there are significant challenges even in the setting with a global controller.

\end{remark}

\section{$QD$-Learning for time-varying directed networks}
\label{sec: directed QD}

In order to develop our resilient $QD$-learning algorithm, we will first need to extend the $QD$-learning algorithm for undirected networks in \cite{kar2013cal}  to time-varying directed networks (in the absence of adversaries); we will thus do this in this section. Consider an underlying graph $\mathcal{G}(t)=(\mathcal{V},\mathcal{E}(t))$ that is time-varying, where $\mathcal{E}(t)\subset\mathcal{V}\times\mathcal{V}$ is the set of edges at time $t$. At time $t$, each agent $v_n$ can obtain information from each neighbor $v_l\in\mathcal{N}_n(t)$, where  $\mathcal{N}_n(t)=\{v_l\in\mathcal{V}|(v_l,v_n)\in\mathcal{E}(t)\}$ is the neighbor set of $v_n$ at time $t$. 

Each agent $v_n\in\mathcal{V}$ maintains a $\mathbb{R}^{|\mathcal{X}\times\mathcal{U}|}$-valued sequence $\{\mathbf{Q}_t^n\}$ with components $Q_{i,u}^n(t)$ and a $\mathbb{R}^{|\mathcal{X}|}$-valued sequence $\{\mathbf{V}_t^n\}$ with components $V_{i}^n(t)$ successively refined as 
\begin{equation}
\label{eq:V}
V_{i}^n(t)=\min_{u\in\mathcal{U}}Q_{i,u}^n(t),\;i=1,2,\cdots,M.  
\end{equation}
Extending the $QD$-learning algorithm from \cite{kar2013cal}, the sequence $\{Q_{i,u}^n(t)\}$ for each state-action pair $(i,u)$ evolves as follows: 
\begin{multline}
\label{eq:QD}
Q_{i,u}^n(t+1)=\\Q_{i,u}^n(t)-\beta_{i,u}(t)\sum_{v_l\in\mathcal{N}_n(t)}\Big(Q_{i,u}^n(t)-Q_{i,u}^l(t)\Big)\\
+\alpha_{i,u}(t)\Big(\!c_n(\mathbf{x}_t,\mathbf{u}_t)\!+\!\gamma \min_{v\in\mathcal{U}}Q_{\mathbf{x}_{t+1},v}^n(t)\!-\!Q_{i,u}^n(t)\!\Big),
\end{multline}
where
\begin{eqnarray}
\label{eq:alpha directed}
\alpha_{i,u}(t)\!\!\!\!&=&\!\!\!\!\left\{
\begin{array}{ll}
a_k, & \mbox{if}\;t=T_{i,u}(k)\;\mbox{for some}\;k\ge 0,\\
\;0, & \mbox{otherwise},\\    
\end{array}\right.\\
\label{eq:beta directed}
\beta_{i,u}(t)\!\!\!\!&=&\!\!\!\!\left\{
\begin{array}{ll}
b, & \mbox{if}\;t=T_{i,u}(k)\;\mbox{for some}\;k\ge 0,\\
\;0, & \mbox{otherwise},
\end{array}\right.
\end{eqnarray}
with $T_{i,u}(k)$ being the $k+1$-th sampling instant of state-action pair $(i,u)$, $a_k\in (0,\eta]$ and $b\in\left[\eta,\frac{1-\eta}{N-1}\right)$ satisfying $\lim\limits_{k\to\infty}a_k=0$, $\sum_{k\ge 0}a_k=\infty$ and $\lim\limits_{k\to\infty}\frac{a_{k-1}}{a_k}=1$, for some constant $\eta\in(0,\frac{1}{N}]$.

\begin{remark}
The update of $Q$-value estimate (\ref{eq:QD}) consists of an innovation term and a consensus term. The innovation term $c_n(\mathbf{x}_t,\mathbf{u}_t)+\gamma \min_{v\in\mathcal{U}}Q_{\mathbf{x}_{t+1},v}^n(t)-Q_{i,u}^n(t)$ is the local $Q$-learning portion. The consensus term $\sum_{v_l\in\mathcal{N}_n(t)}(Q_{i,u}^n(t)-Q_{i,u}^l(t))$ is designed to force all agents to reach consensus on their $Q$-value estimates. 
\end{remark}

\begin{assum}
\label{assum: measurability}
The probability space $(\Omega,\mathcal{F},\mathbb{P})$ is a complete probability space with filtration $\{\mathcal{F}_t\}$ given by $\mathcal{F}_t=\sigma(\{\mathbf{x}_s,\mathbf{u}_s\}_{s\le t},\{c_n(\mathbf{x}_t,\mathbf{u}_t)\}_{v_n\in\mathcal{V},s<t})$. The conditional probability for the controlled transition of $\{\mathbf{x}_t\}$ is $\mathbb{P}(\mathbf{x}_{t+1}=j|\mathcal{F}_t)=p_{\mathbf{x}_tj}^{\mathbf{u}_t}$. For each $v_n$, $\mathbb{E}[c_n(\mathbf{x}_t,\mathbf{u}_t)|\mathcal{F}_t]=\mathbb{E}[c_n(\mathbf{x}_t,\mathbf{u}_t)|\mathbf{x}_t,\mathbf{u}_t]$, which equals $\mathbb{E}[c_n(i,u)]$ on the event $\{\mathbf{x}_t=i,\mathbf{u}_t=u\}$. Further, $c_n(\mathbf{x}_t,\mathbf{u}_t)$ is adapted to $\mathcal{F}_{t+1}$ for each $t$ and $\mathbb{E}[c_n(i,u)]<\infty$.
\end{assum}

\begin{assum}
\label{assum: stopping time}
For each $(i,u)\in\mathcal{X}\times\mathcal{U}$ and each $k\in\mathbb{N}$, the stopping time $T_{i,u}(k)$ is finite a.s., \ie,  $\mathbb{P}(T_{i,u}(k)<\infty)=1$.
\end{assum} 

\begin{defn}[Rooted Graphs]
A graph $\mathcal{G}(t)=(\mathcal{V},\mathcal{E}(t))$ is said to be rooted at node $v_n\in\mathcal{V}$ at time $t$ if for all nodes $v_l\in\mathcal{V}\backslash\{v_n\}$, there is a path from $v_n$ to $v_l$ at time $t$. A path from node $v_n\in\mathcal{V}$ to $v_l\in\mathcal{V}$ is a sequence of nodes $v_{k_1},v_{k_2},\cdots,v_{k_i}$ such that $v_{k_1}=v_n$, $v_{k_i}=v_l$ and $(v_{k_r},v_{k_{r+1}})\in\mathcal{E}(t)$ for $1\le r\le i-1$. A graph $\mathcal{G}(t)=(\mathcal{V},\mathcal{E}(t))$ is said to be rooted at time $t$ if it is rooted at some node $v_n\in\mathcal{V}$ at time $t$.
\end{defn}

\begin{assum}
\label{assum: rooted graph}
The graph $\mathcal{G}(t)=(\mathcal{V},\mathcal{E}(t))$ is directed and rooted for all $t\in\mathbb{N}$.
\end{assum}

For each $v_n$, define the local $QD$-learning operator $\mathcal{G}^n:\mathbb{R}^{|\mathcal{X}\times\mathcal{U}|}\mapsto\mathbb{R}^{|\mathcal{X}\times\mathcal{U}|}$ whose components $\mathcal{G}^n_{i,u}:\mathbb{R}^{|\mathcal{X}\times\mathcal{U}|}\mapsto\mathbb{R}$ are
\[\mathcal{G}_{i,u}^n(\mathbf{Q})=\mathbb{E}[c_n(i,u)]+\gamma\sum_{j\in\mathcal{X}}p_{ij}^u\min_{v\in\mathcal{U}}Q_{j,v}.\]
Let $\mathbf{Q}^{n*}=[Q_{i,u}^{n*}]\in\mathbb{R}^{|\mathcal{X}\times\mathcal{U}|}$ be the fixed point of $\mathcal{G}^n$, \ie, $Q_{i,u}^{n*}$, $\forall (i,u)\in\mathcal{X}\times\mathcal{U}$, satisfy
\[Q_{i,u}^{n*}=\mathbb{E}[c_n(i,u)]+\gamma\sum_{j\in\mathcal{X}}p_{ij}^u\min_{v\in\mathcal{U}}Q_{j,v}^{n*}.\]
Let $\mathbf{V}^{n*}=[V_{i}^{n*}]\in\mathbb{R}^{|\mathcal{X}|}$ be the optimal value function of agent $v_n$, where $V_i^{n*}=\min\limits_{u\in\mathcal{U}}Q_{i,u}^{n*}$.

Define the centralized $Q$-learning operator of all agents $\bar{\mathcal{G}}:\mathbb{R}^{|\mathcal{X}\times\mathcal{U}|}\mapsto\mathbb{R}^{|\mathcal{X}\times\mathcal{U}|}$, whose components $\bar{\mathcal{G}}_{i,u}:\mathbb{R}^{|\mathcal{X}\times\mathcal{U}|}\mapsto\mathbb{R}$ are
\begin{eqnarray*}
\bar{\mathcal{G}}_{i,u}(\mathbf{Q})\!\!\!\!&=&\!\!\!\!\frac{1}{N}\sum_{v_n\in\mathcal{V}}\mathcal{G}_{i,u}^n(\mathbf{Q})\\
&=&\!\!\!\!\frac{1}{N}\sum_{v_n\in\mathcal{V}}\mathbb{E}[c_n(i,u)]+\gamma\sum_{j\in\mathcal{X}}p_{ij}^u\min_{v\in\mathcal{U}}Q_{j,v}.
\end{eqnarray*}
Let $\mathbf{Q}^*=[Q_{i,u}^*]\in\mathbb{R}^{|\mathcal{X}\times\mathcal{U}|}$ be the fixed point of $\bar{\mathcal{G}}$, \ie, $Q_{i,u}^{*}$, $\forall (i,u)\in\mathcal{X}\times\mathcal{U}$, satisfy
\[Q_{i,u}^{*}=\frac{1}{N}\sum_{v_n\in\mathcal{V}}\mathbb{E}[c_n(i,u)]+\gamma\sum_{j\in\mathcal{X}}p_{ij}^u\min_{v\in\mathcal{U}}Q_{j,v}^{*}.\]
Proposition 5.1 in \cite{kar2013cal} indicates that $V_i^{*}=\min\limits_{u\in\mathcal{U}}Q_{i,u}^{*}$.

\subsection{Equivalent Expressions of the $Q$-value Update (\ref{eq:QD})}

Under Assumption \ref{assum: measurability}, equation (\ref{eq:QD}) is equivalent to
\begin{multline}
\label{eq:Q-G-n}
Q_{i,u}^n(t+1)\\
=Q_{i,u}^n(t)-\beta_{i,u}(t)\sum_{v_l\in\mathcal{N}_n(t)}\Big(Q_{i,u}^n(t)-Q_{i,u}^l(t)\Big)\\
+\alpha_{i,u}(t)\Big(\mathcal{G}_{i,u}^n(\mathbf{Q}_t^n)-Q_{i,u}^n(t)+\boldsymbol{\nu}_{\mathbf{x}_t,\mathbf{u}_t}^n(\mathbf{Q}_t^n)\Big),
\end{multline}
where $\boldsymbol{\nu}_{\mathbf{x}_t,\mathbf{u}_t}^n(\mathbf{Q}_t^n)=c_n(\mathbf{x}_t,\mathbf{u}_t)+\gamma \min_{v\in\mathcal{U}}Q_{\mathbf{x}_{t+1},v}^n(t)-\mathcal{G}_{i,u}^n(\mathbf{Q}_t^n)$, satisfying 
$\mathbb{E}[\boldsymbol{\nu}_{\mathbf{x}_t,\mathbf{u}_t}^n(\mathbf{Q}_t^n)|\mathcal{F}_t]=\boldsymbol{0}$ for all $t$. 
Equation (\ref{eq:Q-G-n}) with weights (\ref{eq:alpha directed})-(\ref{eq:beta directed}) is written as 
\begin{eqnarray}
\label{eq:Q-A-n}
Q_{i,u}^n(t+1)\!\!\!\!&=&\!\!\!\!\omega_{i,u}^{nn}(t)Q_{i,u}^n(t)+\sum_{v_l\in\mathcal{N}_n(t)}\omega_{i,u}^{nl}(t)Q_{i,u}^l(t)\nn\\
&&\!\!-\alpha_{i,u}(t)d_{\mathbf{x}_t,\mathbf{u}_t}^n(\mathbf{Q}^n_t),
\end{eqnarray}
where $\omega_{i,u}^{nn}(t)=1-\beta_{i,u}(t)|\mathcal{N}_n(t)|$, $\omega_{i,u}^{nl}(t)=\beta_{i,u}(t),\;v_l\in\mathcal{N}_n(t)$ and $d_{\mathbf{x}_t,\mathbf{u}_t}^n(\mathbf{Q}^n_t)=Q_{i,u}^n(t)-\mathcal{G}_{i,u}^n(\mathbf{Q}_t^n)-\boldsymbol{\nu}_{\mathbf{x}_t,\mathbf{u}_t}^n(\mathbf{Q}_t^n)$.

Let
\[\bar{Q}_{i,u}^n(t)=\mathbb{E}[Q_{i,u}^n(t)|\mathcal{F}_t],\;\forall v_n\in\mathcal{V},\;(i,u)\in\mathcal{X}\times\mathcal{U}.\]
By (\ref{eq:Q-A-n}), $\{\bar{Q}_{i,u}^n(t)\}$ evolves as 
\begin{eqnarray}
\label{eq:Q-A-n-expectation}
\bar{Q}_{i,u}^n(t+1)\!\!\!\!&=&\!\!\!\!\omega_{i,u}^{nn}(t)\bar{Q}_{i,u}^n(t)+\sum_{v_l\in\mathcal{N}_n(t)}\omega_{i,u}^{nl}(t)\bar{Q}_{i,u}^l(t)\nn\\
&&\!\!-\alpha_{i,u}(t)(\bar{Q}_{i,u}^n(t)-\mathcal{G}_{i,u}^n(\bar{\mathbf{Q}}_t^n)),
\end{eqnarray}
where $\bar{\mathbf{Q}}_t^n=\mathbb{E}[\mathbf{Q}_t^n|\mathcal{F}_t]$.

For $k\in\mathbb{N}$, let 
\[z_{i,u}^n(k)=\bar{Q}_{i,u}^n(T_{i,u}(k)),\;\forall v_n\in\mathcal{V},\;(i,u)\in\mathcal{X}\times\mathcal{U}.\] 
Since $\{\bar{Q}_{i,u}^n(t)\}$ only changes at the stopping times $T_{i,u}(k)$, by (\ref{eq:Q-A-n-expectation}), $\{z_{i,u}^n(k)\}$ evolves as 
\begin{eqnarray}
\label{eq:z-A-n}
z_{i,u}^n(k+1)\!\!\!\!&=&\!\!\!\!\hat{\omega}_{i,u}^{nn}(k)z_{i,u}^n(k)+\sum_{v_l\in\mathcal{N}_n^k}\hat{\omega}_{i,u}^{nl}(k)z_{i,u}^l(k)\nn\\
&&\!\!-a_kd_{i,u}^n(\mathbf{z}_k^n),
\end{eqnarray}
where $d_{i,u}^n(\mathbf{z}_k^n)=z_{i,u}^n(k)-\mathcal{G}_{i,u}^n(\mathbf{z}_k^n)$, with $\mathbf{z}_k^n=[z_{i,u}^n(k)]\in\mathbb{R}^{|\mathcal{X}\times\mathcal{U}|}$, $\hat{\omega}_{i,u}^{nl}(k)=b$, $v_l\in\mathcal{N}_n^k$, and $\hat{\omega}_{i,u}^{nn}(k)=1-b|\mathcal{N}_n^k|$, with $\mathcal{N}_n^k=\mathcal{N}_n(T_{i,u}(k))$.

Denote $\mathbf{z}_{i,u}(k)=[z_{i,u}^1(k)\;z_{i,u}^2(k)\;\cdots\;z_{i,u}^N(k)]\T$, $\forall (i,u)\in\mathcal{X}\times\mathcal{U}$. By (\ref{eq:z-A-n}), $\{\mathbf{z}_{i,u}(k)\}$ evolves as 
\begin{equation}
\label{eq:z}
\mathbf{z}_{i,u}(k+1)=A_{i,u}^k\mathbf{z}_{i,u}(k)-a_k\bar{\mathbf{d}}_{i,u}(\mathbf{z}_k).
\end{equation}  
Here, $A_{i,u}^k=I_N-bL^k_{i,u}$ whose $(n,l)$-th entry is $\hat{\omega}_{i,u}^{nl}(k)$ and $\bar{\mathbf{d}}_{i,u}(\mathbf{z}_k)=\mathbf{z}_{i,u}(k)-\mathcal{G}_{i,u}(\mathbf{z}_k)$, where $L^k_{i,u}=L(T_{i,u}(k))$, $\mathcal{G}_{i,u}(\mathbf{z}_k)=[\mathcal{G}_{i,u}^1(\mathbf{z}_k^1)\;\mathcal{G}_{i,u}^2(\mathbf{z}_k^2)\;\cdots\;\mathcal{G}_{i,u}^N(\mathbf{z}_k^N)]\T$, and $\mathbf{z}_k=\left[\mathbf{z}_k^1\;\mathbf{z}_k^2\;\cdots\;\mathbf{z}_k^N\right]$.

\subsection{Convergence of $QD$-Learning}

The proofs of the propositions given in this subsection can be found in the appendix.

The following proposition guarantees the boundedeness of $Q$-value estimates.

\begin{prop}[Boundedness]
	\label{prop: boundedness}
	 Let $\{\mathbf{Q}_t^n\}$ be the successive iterates obtained at agent $v_n$ by (\ref{eq:QD}). Then, under Assumptions \ref{assum: measurability} and \ref{assum: stopping time}, for each agent $v_n\in\mathcal{V}$, $\{\mathbf{Q}_{t}^n\}$ is pathwise bounded, \ie, $\mathbb{P}(\sup_{t\ge 0}\|\mathbf{Q}_{t}^n\|_{\infty}<\infty)=1$.
\end{prop}

Under Assumption \ref{assum: rooted graph}, $A_{i,u}^k$ is rooted for all $k\in\mathbb{N}$. Since $b\in [\eta,\frac{1-\eta}{N-1})$, 
$\hat{\omega}_{i,u}^{nl}(k)$ is lower bounded by $\eta$ for all $k\in\mathbb{N}$. Let $\Phi_{i,u}(k,s)=A_{i,u}^kA_{i,u}^{k-1}\cdots A_{i,u}^s$ for $k\ge s\ge 0$. 
By Lemma 3.4 in \cite{sundaram2018distributed}, for each $s$, there exists a stochastic vector $\mathbf{q}_{i,u}(s)=[q_{i,u}^1(s)\; q_{i,u}^2(s)\;\cdots\;q_{i,u}^N(s)]\T\in\mathbb{R}^N$ 
such that $\lim_{k\to\infty}\Phi_{i,u}(k,s)=\boldsymbol{1}\mathbf{q}_{i,u}\T(s)$. Note that $\mathbf{q}_{i,u}\T(s)=\mathbf{q}_{i,u}\T(s+1)A_{i,u}^s$. Denote by $\{\mathbf{Q}_{i,u}(t)\}$ the $\{\mathcal{F}_t\}$ adapted process with $\mathbf{Q}_{i,u}(t)=[Q_{i,u}^1(t)\;Q_{i,u}^2(t)\;\cdots\;Q_{i,u}^N(t)]\T$. The following proposition establishes the consensus in the agent $Q$-value updates.

\begin{prop}[Consensus]
\label{prop: consensus}
Let $\{\mathbf{Q}_t^n\}$ be the successive iterates obtained at agent $v_n$ by (\ref{eq:QD}). Then, under Assumptions \ref{assum: measurability}-\ref{assum: rooted graph}, agents reach consensus asymptotically, 
\[\mathbb{P}\Big(\limsup_{t\to\infty}\|\mathbf{Q}_{i,u}(t)-\boldsymbol{1}\mathbf{p}_{i,u}\T(t)\mathbf{Q}_{i,u}(t)\|=0\Big)=1,\]
where $\mathbf{p}_{i,u}(t)=\mathbf{q}_{i,u}(k)$, $t\in[T_{i,u}(k),T_{i,u}(k+1))$.
\end{prop}

\begin{prop} 
\label{prop:QD directed}
     Consider the network $\mathcal{G}(t)=(\mathcal{V},\mathcal{E}(t))$. Let $\{\mathbf{Q}_t^n\}$ and $\{\mathbf{V}_t^n\}$ be the successive iterates obtained at agent $v_n$ by the $QD$-learning algorithm (\ref{eq:V})-(\ref{eq:QD}) with weights (\ref{eq:alpha directed})-(\ref{eq:beta directed}). Then, under Assumptions \ref{assum: measurability}-\ref{assum: rooted graph}, for each agent $v_n\in\mathcal{V}$,
    \[\mathbb{P}\Big(\limsup\limits_{t\to\infty}\|\mathbf{Q}_t^n-\mathbf{Q}^{*}\|_\infty\le R\Big)=1,\]
     \[\mathbb{P}\Big(\limsup\limits_{t\to\infty}\|\mathbf{V}_t^n-\mathbf{V}^{*}\|_\infty\le R\Big)=1,\]
     where $R=\max\limits_{v_n,v_l\in\mathcal{V}}\|\mathbf{Q}^{n*}-\mathbf{Q}^{l*}\|_\infty$. For each $i\in\mathcal{X}$, if $|Q_{i,u}^*-Q_{i,v}^*|\ge 2R,\;u,v\in\mathcal{U}$, each agent can learn the optimal policy $\pi^*$.  Furthermore, for each agent $v_n\in\mathcal{V}$ and state-action pair $(i,u)\in\mathcal{X}\times\mathcal{U}$,
    \[\mathbb{P}\big(\limsup\limits_{t\to\infty}Q_{i,u}^n(t)\le M\big)=1,\mathbb{P}\big(\liminf\limits_{t\to\infty}Q_{i,u}^n(t)\ge m\big)=1,\]
   where $M=\max\limits_{v_n\in\mathcal{V}}\max\limits_{i,u}Q_{i,u}^{n*}$, and $m=\min\limits_{v_n\in\mathcal{V}}\min\limits_{i,u}Q_{i,u}^{n*}$.   
\end{prop}

\begin{remark}
 Note that if the matrices $A_{i,u}^k$ do not have a common left-eigenvector, convergence to a constant value is not guaranteed. Thus, the convergence of $\mathbf{Q}_t^n$ to $\mathbf{Q}^*$ cannot be guaranteed for a time-varying directed graph. Instead, Proposition \ref{prop:QD directed} provides estimates of the region of the final consensus value and the distance to the optimal value function $\mathbf{V}^*$.
\end{remark}

\section{Resilient $QD$-Learning}
\label{sec: resilient QD}

With the results on $QD$-learning in time-varying directed graphs in hand, we now turn our attention to analyzing networks with Byzantine adversaries. In this section, we will first show the vulnerability of the $QD$-learning algorithm (\ref{eq:V})-(\ref{eq:QD}) in the presence of a single adversarial agent. After that, we will provide a resilient $QD$-learning algorithm  that can handle a potentially large number of adversaries.  We start with the following definitions.

\begin{defn}[\cite{leblanc2013resilient} Byzantine agent]
A Byzantine agent is capable of behaving arbitrarily (\ie, it may not follow the prescribed algorithms), and is allowed to send conflicting or incorrect values to different neighbors at each time-step. It is also allowed to know the network topology and the private information of all other agents.
\end{defn}

\begin{defn}[\cite{leblanc2013resilient} $r$-reachable set] 
Consider a graph $\mathcal{G}=(\mathcal{V},\mathcal{E})$. For any given $r\in\mathbb{N}$, a subset of nodes $\mathcal{S}\subseteq\mathcal{V}$ is said to be $r$-reachable if there exists a node $v_n\in\mathcal{S}$ such that $|\mathcal{N}_n\setminus\mathcal{S}|\ge r$.  
\end{defn}

\begin{defn}[\cite{leblanc2013resilient} $r$-robust graphs]
For $r\in\mathbb{N}$, graph $\mathcal{G}$ is said to be $r$-robust if for all pairs of disjoint nonempty subsets $S_1,S_2\subset\mathcal{V}$, at least one of $S_1$ or $S_2$ is $r$-reachable. 
\end{defn}

\begin{defn}[\cite{leblanc2013resilient} F-local set]
For $F\in\mathbb{N}$, the set of adversaries $\mathcal{A}$ is an $F$-local set if $|\mathcal{N}_n\cap\mathcal{A}|\le F$, for all $v_n\in\mathcal{R}$. 
\end{defn}

\begin{assum}
\label{assum: Byzantine}
The adversarial nodes are Byzantine agents and restricted to form a $F$-local set. 
The agent network $\mathcal{G}=\{\mathcal{V},\mathcal{E}\}$ is time-invariant and $(2F+1)$-robust. 
\end{assum}

The following proposition illustrates that by running the $QD$-learning algorithm, regular agents cannot learn the optimal value function and the optimal policy even in the presence of a single adversarial agent. 

\begin{prop}
\label{prop: QD failure}
Consider the time-invariant network $\mathcal{G}=(\mathcal{V},\mathcal{E})$, and let there be a single adversarial node $\mathcal{A}=\{v_N\}$. Suppose the network is connected, and all agents run the $QD$-learning algorithm (\ref{eq:V})-(\ref{eq:QD}). If $v_N$ keeps its $Q$-value estimate $Q_{i,u}^N(t)$ fixed at some arbitrary value $Q_{i,u}^{N*}$, for each regular agent $v^n\in\mathcal{R}$, $Q_{i,u}^n(t)\to Q_{i,u}^{N*}$ and $V_{i}^n(t)\to V_{i}^{N*}$ as $t\to\infty$ \as.
\end{prop}
\begin{proof}
Since the adversarial node keeps its value fixed for all time, $\{Q_{i,u}^N(t)\}$ is updated as $Q_{i,u}^N(t+1)=Q_{i,u}^{N}(t)$, for all $t\in\mathbb{N}$, with $Q_{i,u}^{N}(0)=Q_{i,u}^{N*}$. Thus, the dynamics of $\mathbf{z}_{i,u}(k)$ take the form of (\ref{eq:z}), with
\[ A_{i,u}^k=\left[\begin{array}{cc}
A_{i,u}^{\mathcal{R},\mathcal{R}}(k)    &  A_{i,u}^{\mathcal{R},\mathcal{A}}(k) \\
0     & 1
\end{array}\right],\]
where $A_{i,u}^{\mathcal{R},\mathcal{R}}(k) =[\hat{\omega}_{i,u}^{nl}(k)]\in\mathbb{R}^{N-1\times N-1}$ contains the weights placed by regular agents on other regular agents, and $A_{i,u}^{\mathcal{R},\mathcal{A}}(k)=[\hat{\omega}_{i,u}^{1N}(k)\;\hat{\omega}_{i,u}^{2N}(k)\;\cdots\;\hat{\omega}_{i,u}^{NN}(k)]\T\in\mathbb{R}^{N}$. For all $k\in\mathbb{N}$, $A_{i,u}^k$ have a common left-eigenvector $\mathbf{q}\T=[0_{1\times N-1}\;1]$. Then, by Proposition \ref{prop: consensus}, $z_{i,u}^n(k)$ will converge to $\mathbf{q}\T \mathbf{z}_{i,u}(k)=z_{i,u}^N(k)=Q_{i,u}^{N*}$, which indicates $Q_{i,u}^n(t)$, $\forall v_n\in\mathcal{R}$, will converge to $Q_{i,u}^{N*}$ \as.
\end{proof}

The following proposition illustrates that {\it any} algorithm that always finds the optimal value function and the optimal policy in the absence of adversaries can also be arbitrarily co-opted by an adversary.

\begin{prop}
\label{prop: co-opted}
Suppose $\Gamma$ is an algorithm that guarantees that all agents learn the optimal value function $\mathbf{V}^*$ and the optimal policy $\pi^*$ when there are no adversarial agents. Then a single adversary can cause all agents to converge to any arbitrary value when running algorithm $\Gamma$, and furthermore, will remain undetected.
\end{prop}

\begin{proof}
Assume $v_N$ is an adversarial agent. Suppose agent $v_N$ wishes all agents to calculate $\mathbf{V}^{N*}$ as an outcome of running the algorithm $\Gamma$. Agent $v_N$ chooses a cost function $\bar{c}_N(i,u)=-\sum_{v_n\in\mathcal{V}\setminus\{v_N\}}c_n(i,u)+c_N(i,u)$. Now agent $v_N$ participates in algorithm $\Gamma$ by pretending its local cost function is $\bar{c}_N(i,u)$ instead of $c_N(i,u)$. Since $\bar{c}_N(i,u)$ is a legitimate cost that could have been assigned to $v_N$, this scenario is indistinguishable from the cases that where $v_N$ is a regular agent. Thus, algorithm $\Gamma$ must cause all agents to learn $\mathbf{V}^{N*}$.
\end{proof}

The above results show that the price for resilience is a loss of optimality (in general). This motivates us to create a resilient algorithm that provides approximately optimal solutions. To do this, consider a modification of the $QD$-learning algorithm, where each regular agent $v_n$ updates $Q_{i,u}^n(t)$ for state-action pair $(i,u)$ as
\begin{multline}
    \label{eq:Q-resilient}
       Q_{i,u}^n(t+1)\\
       =Q_{i,u}^n(t)-\beta_{i,u}(t)\sum_{v_l\in\mathcal{J}_{i,u}^n(t)}(Q_{i,u}^n(t)-Q_{i,u}^l(t))\\
    +\alpha_{i,u}(t)\big(c_n(\mathbf{x}_t,\mathbf{u}_t)+\gamma \min_{v\in\mathcal{U}}Q_{\mathbf{x}_{t+1},v}^n(t)-Q_{i,u}^n(t)\big),
\end{multline}
where $\alpha_{i,u}(t)$ and $\beta_{i,u}(t)$ are in (\ref{eq:alpha directed}) and (\ref{eq:beta directed}), and $\mathcal{J}^n_{i,u}(t)\in\mathcal{N}_n$ is computed by the following procedure. Agent $v_n$ receives $\{Q_{i,u}^l(t)$, $l\in\mathcal{N}_n\}$ and removes the $F$ highest and $F$ smallest values that are larger and smaller than $Q_{i,u}^n(t)$, respectively. If there are fewer than $F$ values higher than $Q_{i,u}^n(t)$, agent $v_n$ removes all values that are strictly larger than $Q_{i,u}^n(t)$. Likewise, if there are less than $F$ values strictly smaller than $Q_{i,u}^n(t)$, then agent $v_n$ removes all values that are strictly smaller than $Q_{i,u}^n(t)$. Otherwise, it removes precisely the smallest $F$ values. Let $\mathcal{J}^n_{i,u}(t)\in\mathcal{N}_n$ denote the set of agents whose values were retained by regular agent $v_n$ at time $t$ for state-action pair $(i,u)$.

The above resilient $QD$-Learning algorithm for each regular agent $v_n\in\mathcal{R}$ is summarized in Algorithm \ref{algorithm: RQD}.
\begin{algorithm}
\caption{Resilient $QD$ Learning Algorithm}
\label{algorithm: RQD}
\begin{algorithmic}[1]
     \State Initialize $\mathbf{Q}_0^n$, $\mathbf{V}_0^n$
     \For{$t=0,1,2,\cdots$}
         \State Receive state $\mathbf{x}_{t}$, action $\mathbf{u}_t$ and cost $c_n(\mathbf{x_t},\mathbf{u_t})$
         \State Receive state $\mathbf{x}_{t+1}$ and $\mathbf{Q}_{t}^l$, $l\in\mathcal{N}_n$
         \For{$(i,u)\in\mathcal{X}\times\mathcal{U}$}
       \State Compute $\mathcal{J}^n_{i,u}(t)\in\mathcal{N}_n$
       \State Compute $Q_{i,u}^n(t+1)$ as (\ref{eq:Q-resilient})
      \EndFor
     \For{$i\in\mathcal{X}$}
     \State Compute $V_{i}^n(t+1)=\min\limits_{u\in\mathcal{U}}Q_{i,u}^n(t+1)$
     \EndFor  
     \EndFor
\end{algorithmic}
\end{algorithm}

Define the centralized $Q$-learning operator of all regular agents  $\bar{\mathcal{G}}^{\mathcal{R}}:\mathbb{R}^{|\mathcal{X}\times\mathcal{U}|}\mapsto\mathbb{R}^{|\mathcal{X}\times\mathcal{U}|}$, whose components $\bar{\mathcal{G}}_{i,u}^{\mathcal{R}}:\mathbb{R}^{|\mathcal{X}\times\mathcal{U}|}\mapsto\mathbb{R}$ are
\begin{eqnarray*}
\bar{\mathcal{G}}_{i,u}^{\mathcal{R}}(\mathbf{Q})\!\!\!\!&=&\!\!\!\!\frac{1}{|\mathcal{R}|}\sum_{v_n\in\mathcal{R}}\mathcal{G}_{i,u}^n(\mathbf{Q})\\
&=&\!\!\frac{1}{|\mathcal{R}|}\sum_{v_n\in\mathcal{R}}\mathbb{E}[c_n(i,u)]+\gamma\sum_{j\in\mathcal{X}}p_{ij}^u\min_{v\in\mathcal{U}}Q_{j,v}.
\end{eqnarray*}
Let $\mathbf{Q}^{\mathcal{R}*}=[Q_{i,u}^{\mathcal{R}*}]\in\mathbb{R}^{|\mathcal{X}\times\mathcal{U}|}$ be the fixed point of $\bar{\mathcal{G}}^{\mathcal{R}}$, 
\ie, $Q_{i,u}^{\mathcal{R}*}$, $\forall (i,u)\in\mathcal{X}\times\mathcal{U}$, satisfy
\[Q_{i,u}^{\mathcal{R}*}=\frac{1}{|\mathcal{R}|}\sum_{v_n\in\mathcal{R}}\mathbb{E}[c_n(i,u)]+\gamma\sum_{j\in\mathcal{X}}p_{ij}^u\min_{v\in\mathcal{U}}Q_{j,v}^{\mathcal{R}*}.\]
Let $\mathbf{V}^{\mathcal{R}*}=[V_{i}^{\mathcal{R}*}]\in\mathbb{R}^{|\mathcal{X}|}$ be the optimal value function of all regular agents, where $V_i^{\mathcal{R}*}=\min\limits_{u\in\mathcal{U}}Q_{i,u}^{\mathcal{R}*}$.

We will use the following result in our analysis of Algorithm \ref{algorithm: RQD}.

\begin{lem}[\cite{vaidya2012matrix,sundaram2018distributed}]
\label{lem: equivalent}
Consider a network $\mathcal{G}=(\mathcal{V},\mathcal{E})$, with a set of regular nodes $\mathcal{R}$ and a set of adversarial nodes $\mathcal{A}$. Suppose that $\mathcal{A}$ is an $F$-local set, and that each regular node has at least $2F+1$ neighbors. Consider an iteration of the form
\begin{eqnarray}
\label{x:DO}
x_n(k+1)\!\!\!\!&=&\!\!\!\!a_{nn}(k)x_n(k)+\sum_{v_l\in\mathcal{J}^n(k)}\!\!\!\!a_{nl}(k)x_l(k)\nn\\
&&\!\!\!\!-a_kd_n(k),  
\end{eqnarray}
where $a_{nl}(k)\ge\eta$, $\sum_{l}a_{nl}(k)=1$, $v_l\in\{v_n\}\cup\mathcal{J}^n(k)$, with $J^n(k)$ being generated in the same way as $J_{i,u}^n(t)$ and $d_n(k)$ is a given sequence. 
Equation (\ref{x:DO}) is equivalent to 
\[x_n(k+1)=\bar{a}_{nn}(k)x_n(k)+\sum_{v_l\in\mathcal{N}_n\cap\mathcal{R}}\bar{a}_{nl}(k)x_l(k)-a_kd_n(k),\]
where the weights $\bar{a}_{nl}(k)$ are nonnegative and satisfy the following properties:
\begin{itemize}
    \item $\bar{a}_{nn}(k)+\sum\limits_{v_l\in\mathcal{N}_n\cap\mathcal{R}}\bar{a}_{nl}(k)=1$,
    \item $\bar{a}_{nn}(k)\ge\eta$ and at least $|\mathcal{N}_n|-2F$ of other weights are lower bounded by $\frac{\eta}{2}$.
\end{itemize}
\end{lem}

We now come to the main result in our paper.
\begin{thm}
\label{thm-resilient}
    Consider the network $\mathcal{G}=(\mathcal{V},\mathcal{E})$ with regular nodes $\mathcal{R}$ and adversarial nodes $\mathcal{A}$.   Under Assumptions \ref{assum: measurability}, \ref{assum: stopping time} and \ref{assum: Byzantine}, Algorithm \ref{algorithm: RQD} guarantees that, for each regular agent $v_n\in\mathcal{R}$, \[\mathbb{P}\Big(\limsup\limits_{t\to\infty}\left\|\mathbf{Q}_t^n-\mathbf{Q}^{\mathcal{R}*}\right\|_\infty\le R\Big)=1,\]
     \[\mathbb{P}\Big(\limsup\limits_{t\to\infty}\|\mathbf{V}_t^n-\mathbf{V}^{\mathcal{R}*}\|_\infty\le R\Big)=1,\]
     where 
    \begin{equation}
    \label{RQD:R}
        R=\max_{v_n,v_l\in\mathcal{R}}\|\mathbf{Q}^{n*}-\mathbf{Q}^{l*}\|_\infty.
    \end{equation}
     For each $i\in\mathcal{X}$, if $|Q_{i,u}^{\mathcal{R}*}-Q_{i,v}^{\mathcal{R}*}|\ge 2R$, $u,v\in\mathcal{U}$, each regular agent can learn the optimal  policy $\pi^{\mathcal{R}*}$.  Furthermore, for each regular agent $v_n\in\mathcal{R}$ and state-action pair $(i,u)\in\mathcal{X}\times\mathcal{U}$,
     \begin{eqnarray}
     \label{RQD:upper bound}
     &&\!\!\!\!\mathbb{P}\Big(\limsup\limits_{t\to\infty}Q_{i,u}^n(t)\le M^{\mathcal{R}}\Big)=1,\\
    \label{RQD:lower bound}
     &&\!\!\!\!\mathbb{P}\Big(\liminf\limits_{t\to\infty}Q_{i,u}^n(t)\ge m^{\mathcal{R}}\Big)=1,
     \end{eqnarray}
     where $M^{\mathcal{R}}=\max\limits_{v_n\in\mathcal{R}}\max\limits_{i,u}Q_{i,u}^{n*},$ and $m^{\mathcal{R}}=\min\limits_{v_n\in\mathcal{R}}\min\limits_{i,u}Q_{i,u}^{n*}$.
\end{thm}

\begin{proof}
By (\ref{eq:Q-resilient}), $\{z_{i,u}^n(k)\}$ evolves as
 \begin{eqnarray}
    \label{eq:z-resilient}
        z_{i,u}^n(k+1)\!\!\!\!&=&\!\!\!\!\hat{\omega}_{i,u}^{nn}(k)z_{i,u}^n(k)+\!\!\!\!\sum_{v_l\in\mathcal{J}_{i,u}^n(T_{i,u}(k))}\hat{\omega}_{i,u}^{nl}(k)z_{i,u}^l(k)\nn\\
        &&\!\!\!\!-a_kd_{i,u}^n(\mathbf{z}^n_k),
    \end{eqnarray}
where $\hat{\omega}_{i,u}^{nn}(k)=1-b|\mathcal{J}_{i,u}^n(T_{i,u}(k))|$, $\hat{\omega}_{i,u}^{nl}(k)=b$, $v_l\in\mathcal{N}_n$ and $d_{i,u}^n(\mathbf{z}^n_k)=z_{i,u}^n(k)-\mathcal{G}_{i,u}^n(\mathbf{z}_k^n)$ with $\mathbf{z}_k^n\in\mathbb{R}^{|\mathcal{X}\times\mathcal{U}|}$ whose components are $z_{i,u}^n(k)$. 

By Lemma \ref{lem: equivalent}, the update rule (\ref{eq:z-resilient}) for each $v_n\in\mathcal{R}$ is equivalent to
 \begin{eqnarray}
    \label{eq:z-resilient-equivalent}
        z_{i,u}^n(k+1)\!\!\!\!&=&\!\!\!\!\bar{\omega}_{i,u}^{nn}(k)z_{i,u}^n(k)+\sum_{v_l\in\mathcal{N}_n\cap\mathcal{R}}\bar{\omega}_{i,u}^{nl}(k)z_{i,u}^l(k)\nn\\
        &&\!\!\!\!-a_kd_{i,u}^n(\mathbf{z}^n_k),
    \end{eqnarray}
where the weights $\bar{\omega}_{i,u}^{nl}(k)$ are nonnegative and satisfy the following properties:
\begin{itemize}
    \item $\bar{\omega}_{i,u}^{nn}(k)+\sum\limits_{v_l\in\mathcal{N}_n\cap\mathcal{R}}\bar{\omega}_{nl}(k)=1$,
    \item $\bar{\omega}_{i,u}^{nn}(k)\ge\eta$ and at least $|\mathcal{N}_n|-2F$ of other weights are lower bounded by $\frac{\eta}{2}$.
\end{itemize}
Without loss of generality, we assume that the regular nodes are arranged first in the ordering of the nodes.
Let $\mathbf{z}_{i,u}^\mathcal{R}(k)=[z_{i,u}^1(k)\;\cdots\;z_{i,u}^{|\mathcal{R}|}(k)]\T$. Then, we have 
\begin{equation}
\label{eq:Q-A-resilient}
\mathbf{z}_{i,u}^\mathcal{R}(k+1)=\bar{A}_{i,u}(k)\mathbf{z}_{i,u}^\mathcal{R}(k)-a_k\mathbf{d}_{i,u}^\mathcal{R}(\mathbf{z}_k^\mathcal{R}),
\end{equation}
where $\bar{A}_{i,u}(k)\in\mathbb{R}^{|\mathcal{R}|\times|\mathcal{R}|}$ is a matrix whose $(n,l)$-th entry is $\bar{\omega}_{i,u}^{nl}(k)$ and $\mathbf{d}_{i,u}^\mathcal{R}(\mathbf{z}_k^\mathcal{R})=[d_{i,u}^1(\mathbf{z}^1_k)\;\cdots\;d_{i,u}^{|\mathcal{R}|}(\mathbf{z}^{|\mathcal{R}|}_k)]\T$.

Consider the graph $\mathcal{G}$, and remove all edges whose weights are smaller than $\frac{\eta}{2}$ in $\bar{A}_{i,u}(k)$. By Lemma 2.3 in \cite{sundaram2018distributed}, the subgraph consisting of regular nodes will be rooted after removing $2F$ or fewer edges from each regular nodes if the graph is $(2F+1)$-robust. Thus, $\bar{A}_{i,u}(k)$ is rooted for each $k\in\mathbb{N}$, with a tree whose edge-weights are all lower-bounded by $\frac{\eta}{2}$. Thus, equation (\ref{eq:Q-A-resilient}) is in the same form of equation (\ref{eq:z}). Following the same steps in the proof of Proposition \ref{prop:QD directed}, we can establish Theorem \ref{thm-resilient}.
\end{proof}

\begin{remark}
 Regardless of the behavior of any set of Byzantine agents, the error between the value function $\mathbf{V}_t^n$ of each regular agent $v_n$ and the optimal value function $\mathbf{V}^*$ can be further bounded by $R\le\max_{v_n,v_l\in\mathcal{R}}\frac{1}{1-\gamma}\|\mathbb{E}[\mathbf{c}_n]-\mathbb{E}[\mathbf{c}_l]\|_\infty$, where $\mathbf{c}_n=[c_n(i,u)]\in\mathbb{R}^{|\mathcal{X}\times\mathcal{U}|}$. Roughly speaking, $R$ becomes smaller as the local optimal value functions/costs of regular agents get closer. In particular, if all regular agents own the same local optimal value functions/costs, $R=0$. 
\end{remark}

\begin{remark}
Equations (\ref{RQD:upper bound}) and (\ref{RQD:lower bound}) further imply 
\[\mathbb{P}\big(\limsup\limits_{t\to\infty}\|\mathbf{Q}^n_t\|_\infty\le\max_{v_n\in\mathcal{R}}\|\mathbf{Q}^{n*}\|_\infty\big)=1,\;\forall v_n\in\mathcal{R}.\]
More specifically, unlike standard (optimal) distributed learning algorithms that can be arbitrarily co-opted by an adversary (Proposition \ref{prop: co-opted}), in the long run, the $Q$ values of each regular agent will be bounded by the largest maximum norm of local optimal $Q$ values among all regular agents under our algorithm, regardless of the behaviors of any F-local set of Byzantine agents, 
\end{remark}

\begin{remark}
The adversary model we consider is the $F$-local model, which is more general than the  $F$-total model considered in \cite{wu2020byzantine}. In particular, the $F$-total model indicates that there are no more than $F$ Byzantine nodes in the {\it entire network}, whereas we allow up to $F$ Byzantine nodes in the neighborhood of {\it every regular node}.
\end{remark}

\section{Conclusion}
\label{sec: conclusion}
We developed a resilient distributed $Q$-learning algorithm for a networked system in the presence of Byzantine agents. Under certain conditions on the network topology, we established the almost sure convergence of the value function of each regular agent to the neighborhood of the optimal value function of all regular agents. For each state, if the optimal $Q$-values corresponding to different actions are sufficiently separated, our algorithm allows each regular agent to learn the optimal policy of all regular agents.


\bibliography{rqd}

\begin{thebibliography}{10}
\providecommand{\url}[1]{#1}
\csname url@samestyle\endcsname
\providecommand{\newblock}{\relax}
\providecommand{\bibinfo}[2]{#2}
\providecommand{\BIBentrySTDinterwordspacing}{\spaceskip=0pt\relax}
\providecommand{\BIBentryALTinterwordstretchfactor}{4}
\providecommand{\BIBentryALTinterwordspacing}{\spaceskip=\fontdimen2\font plus
\BIBentryALTinterwordstretchfactor\fontdimen3\font minus
  \fontdimen4\font\relax}
\providecommand{\BIBforeignlanguage}[2]{{%
\expandafter\ifx\csname l@#1\endcsname\relax
\typeout{** WARNING: IEEEtran.bst: No hyphenation pattern has been}%
\typeout{** loaded for the language `#1'. Using the pattern for}%
\typeout{** the default language instead.}%
\else
\language=\csname l@#1\endcsname
\fi
#2}}
\providecommand{\BIBdecl}{\relax}
\BIBdecl

\bibitem{kar2013cal}
S.~Kar, J.~M. Moura, and H.~V. Poor, ``$\mathcal{QD}$-learning: A collaborative
  distributed strategy for multi-agent reinforcement learning through consensus
  $+$ innovations,'' \emph{IEEE Transactions on Signal Processing}, vol.~61,
  no.~7, pp. 1848--1862, 2013.

\bibitem{zhang2018fully}
K.~Zhang, Z.~Yang, H.~Liu, T.~Zhang, and T.~Basar, ``Fully decentralized
  multi-agent reinforcement learning with networked agents,'' in
  \emph{International Conference on Machine Learning}.\hskip 1em plus 0.5em
  minus 0.4em\relax PMLR, 2018, pp. 5872--5881.

\bibitem{qu2020scalable}
G.~Qu, A.~Wierman, and N.~Li, ``Scalable reinforcement learning of localized
  policies for multi-agent networked systems,'' in \emph{Learning for Dynamics
  and Control}.\hskip 1em plus 0.5em minus 0.4em\relax PMLR, 2020, pp.
  256--266.

\bibitem{lin2020distributed}
Y.~Lin, G.~Qu, L.~Huang, and A.~Wierman, ``Distributed reinforcement learning
  in multi-agent networked systems,'' \emph{arXiv preprint arXiv:2006.06555},
  2020.

\bibitem{sundaram2018distributed}
S.~Sundaram and B.~Gharesifard, ``Distributed optimization under adversarial
  nodes,'' \emph{IEEE Transactions on Automatic Control}, vol.~64, no.~3, pp.
  1063--1076, 2018.

\bibitem{leblanc2013resilient}
H.~J. LeBlanc, H.~Zhang, X.~Koutsoukos, and S.~Sundaram, ``Resilient asymptotic
  consensus in robust networks,'' \emph{IEEE Journal on Selected Areas in
  Communications}, vol.~31, no.~4, pp. 766--781, 2013.

\bibitem{pasqualetti2012consensus}
F.~Pasqualetti, A.~Bicchi, and F.~Bullo, ``Consensus computation in unreliable
  networks: A system theoretic approach,'' \emph{IEEE Transactions on Automatic
  Control}, vol.~1, no.~57, pp. 90--104, 2012.

\bibitem{wang2019resilient}
X.~Wang, S.~Mou, and S.~Sundaram, ``A resilient convex combination for
  consensus-based distributed algorithms,'' \emph{Numerical Algebra, Control \&
  Optimization}, vol.~9, no.~3, pp. 269--281, 2019.

\bibitem{zhao2019resilient}
C.~Zhao, J.~He, and Q.-G. Wang, ``Resilient distributed optimization algorithm
  against adversarial attacks,'' \emph{IEEE Transactions on Automatic Control},
  vol.~65, no.~10, pp. 4308--4315, 2019.

\bibitem{kuwaranancharoen2020byzantine}
K.~Kuwaranancharoen, L.~Xin, and S.~Sundaram, ``Byzantine-resilient distributed
  optimization of multi-dimensional functions,'' in \emph{2020 American Control
  Conference (ACC)}.\hskip 1em plus 0.5em minus 0.4em\relax IEEE, 2020, pp.
  4399--4404.

\bibitem{chen2017distributed}
Y.~Chen, L.~Su, and J.~Xu, ``Distributed statistical machine learning in
  adversarial settings: Byzantine gradient descent,'' \emph{Proceedings of the
  ACM on Measurement and Analysis of Computing Systems}, vol.~1, no.~2, pp.
  1--25, 2017.

\bibitem{blanchard2017byzantine}
P.~Blanchard, E.~M.~E. Mhamdi, R.~Guerraoui, and J.~Stainer,
  ``Byzantine-tolerant machine learning,'' \emph{arXiv preprint
  arXiv:1703.02757}, 2017.

\bibitem{yin2018byzantine}
D.~Yin, Y.~Chen, R.~Kannan, and P.~Bartlett, ``Byzantine-robust distributed
  learning: Towards optimal statistical rates,'' in \emph{International
  Conference on Machine Learning}.\hskip 1em plus 0.5em minus 0.4em\relax PMLR,
  2018, pp. 5650--5659.

\bibitem{yang2019byrdie}
Z.~Yang and W.~U. Bajwa, ``Byrdie: Byzantine-resilient distributed coordinate
  descent for decentralized learning,'' \emph{IEEE Transactions on Signal and
  Information Processing over Networks}, vol.~5, no.~4, pp. 611--627, 2019.

\bibitem{yang2019bridge}
------, ``Bridge: Byzantine-resilient decentralized gradient descent,''
  \emph{arXiv preprint arXiv:1908.08098}, 2019.

\bibitem{lin2020toward}
Y.~Lin, S.~Gade, R.~Sandhu, and J.~Liu, ``Toward resilient multi-agent
  actor-critic algorithms for distributed reinforcement learning,'' in
  \emph{2020 American Control Conference (ACC)}.\hskip 1em plus 0.5em minus
  0.4em\relax IEEE, 2020, pp. 3953--3958.

\bibitem{wu2020byzantine}
Z.~Wu, H.~Shen, T.~Chen, and Q.~Ling, ``Byzantine-resilient decentralized td
  learning with linear function approximation,'' \emph{arXiv preprint
  arXiv:2009.11146}, 2020.

\bibitem{vaidya2012matrix}
N.~Vaidya, ``Matrix representation of iterative approximate {Byzantine}
  consensus in directed graphs,'' \emph{arXiv preprint arXiv:1203.1888}, 2012.

\end{thebibliography}
\bibliographystyle{IEEEtran}

\appendix
\subsection{Preliminary Lemmas}
\begin{lem}
	\label{lem convergence}
	For each state-action pair $(i,u)$, let $\{\mathbf{y}_{i,u}(t)\}$ denote the $\{\mathcal{F}_t\}$ adapted process evolving as 
\begin{eqnarray}
\mathbf{y}_{i,u}(t+1)\!\!\!\!&=&\!\!\!\!(I_N-\beta_{i,u}(t)L(t)-\alpha_{i,u}(t)I_N)\mathbf{y}_{i,u}(t)\nn\\
&&\!\!\!\!+\alpha_{i,u}(t)\bar{\boldsymbol{\nu}}_{i,u}(t),
\end{eqnarray}
where the weighted sequences $\{\alpha_{i,u}\}$ and $\{\beta_{i,u}\}$ are given by (\ref{eq:alpha directed}) and (\ref{eq:beta directed}) and $\{\bar{\boldsymbol{\nu}}_{i,u}(t)\}$ is an $\{\mathcal{F}_{t+1}\}$ adapted process satisfying $\mathbb{E}[\bar{\boldsymbol{\nu}}_{i,u}(t)|\mathcal{F}_t]=\boldsymbol{0}$ for all $t$. Then, under Assumption \ref{assum: stopping time}, we have $\mathbf{y}_{i,u}\to 0$ as $t\to\infty$ a.s..
\end{lem}
\begin{proof}
For each $k\ge 0$, let $\mathcal{H}_{i,u}^k$ be the $\sigma$-algebra associated with the stopping time $T_{i,u}(k)$, \ie, $\mathcal{H}_{i,u}^k=\mathcal{F}_{T_{i,u}(k)}$. Let $\{\mathbf{z}_{i,u}(k)\}$ denote the randomly sampled version of $\{\mathbf{y}_{i,u}(t)\}$, \ie, $\mathbf{z}_{i,u}(k)=\mathbf{y}_{i,u}(T_{i,u}(k))$, $\forall k$.
The process $\{\mathbf{z}_{i,u}(k)\}$ is $\{\mathcal{H}_{i,u}^k\}$ adapted, which evolves as 
\[ \mathbf{z}_{i,u}(k+1)=A_{i,u}^k\mathbf{z}_{i,u}(k)-a_k\mathbf{z}_{i,u}(k)+a_k\bar{\boldsymbol{\nu}}_k.  \]
Then,
\[\mathbb{E}[\mathbf{z}_{i,u}(k+1)|\mathcal{H}_{i,u}^k]=\mathbb{E}[A_{i,u}^k-a_k I]\mathbb{E}[\mathbf{z}_{i,u}(k)|\mathcal{H}_{i,u}^k],\]
and
\[\|\mathbb{E}[\mathbf{z}_{i,u}(k+1)|\mathcal{H}_{i,u}^k]\|_{\infty}\le\prod_{s=0}^{k}(1-a_k)\|\mathbb{E}[\mathbf{z}_{i,u}(0)|\mathcal{H}_{i,u}^k]\|_\infty.\]
Since $\sum_{k\ge 0}a_k=\infty$, we obtain that $\lim_{k\to\infty}\prod_{s=0}^{k}(1-a_k)\le 0$, and, hence, $\lim_{k\to\infty}\|\mathbb{E}[\mathbf{z}_{i,u}(k)|\mathcal{H}_{i,u}^k]\|=0$, which indicates $\lim_{k\to\infty}\mathbf{z}_{i,u}(k)=0$ \as. Since $\{\mathbf{y}_{i,u}(t)\}$ is a piecewise constant interpolation of $\{\mathbf{z}_{i,u}(k)\}$, we obtain $\mathbb{P}(\lim_{t\to\infty}\|\mathbf{y}_{i,u}(t)\|=0)=1$. 
\end{proof}

\begin{lem}
\label{lem: coro}
	For each state-action pair $(i,u)$ and $t_0\ge 0$, consider the process $\{\mathbf{z}_{i,u}(t:t_0)\}_{t\ge t_0}$ that evolves as 
	\begin{eqnarray*}
	\mathbf{z}_{i,u}(t+1:t_0)\!\!\!\!&=&\!\!\!\!(I_N-\beta_{i,u}(t)L(t)-\alpha_{i,u}(t)I_N)\mathbf{z}_{i,u}(t:t_0)\\
    &&\!\!\!\!+\alpha_{i,u}(t)\bar{\boldsymbol{\nu}}_{i,u}(t), 
	\end{eqnarray*}
	with $\mathbf{z}_{i,u}(t_0:t_0)=\boldsymbol{0}$, where $\alpha_{i,u}(t)$, $\beta_{i,u}(t)$	and $\bar{\boldsymbol{\nu}}_{i,u}(t)$ satisfy the hypothesis of Lemma \ref{lem convergence}. Then, for each $\varepsilon>0$, there exists a random time $t_{\varepsilon}$ such that $\|\mathbf{z}_{i,u}(t:t_0)\|_\infty\le\varepsilon$, $t_\varepsilon\le t_0\le t$.
\end{lem}

\begin{proof}
Since $1-\alpha_{i,u}(t)-\beta_{i,u}(t)|\mathcal{N}_n(t)|\ge 1-a_k-b(N-1)>0$, the matrix $S_{i,u}(t):=I_N-\beta_{i,u}(t)L(t)-\alpha_{i,u}(t)I_N$ is nonnegative. Thus, $\|S_{i,u}(t)\|_\infty=1-\alpha_{i,u}(t)$, $\forall t$.  
Then, for each $t\ge t_0$,
\begin{multline*}
 \|\mathbf{z}_{i,u}(t:t_0)\|_\infty=\Big\|\mathbf{z}_{i,u}(t:0)-\prod_{s=t_0}^{t-1}S_{i,u}(s)\mathbf{z}_{i,u}(t_0:0)\Big\|_\infty\\\le\left\|\mathbf{z}_{i,u}(t:0)\|_\infty+\|\mathbf{z}_{i,u}(t_0:0)\right\|_\infty.
\end{multline*}
By Lemma \ref{lem convergence}, $\mathbf{z}_{i,u}(t:0)\to\boldsymbol{0}$ as $t\to\infty$ a.s., and, hence, there exists $t_\varepsilon$ such that $\|\mathbf{z}_{i,u}(t:0)\|_\infty\le\frac{\varepsilon}{2}$ for $t\ge t_\varepsilon$. The result follows immediately. 
\end{proof}

\begin{lem}
\label{lem lower bound}
    Let $\{z_t\}$ be a real-valued and deterministic sequence with $z_{t+1}\ge(1-\alpha_t)z_t+\alpha_t\varepsilon_t$,
    where $\alpha_t\in(0,\eta]$ for all $t$ and $\eta\in(0,1)$, $\sum_{t\ge 0}\alpha_t=\infty$, and there exists a constant $R>0$ such that $\liminf_{t\to\infty}\varepsilon_t\ge R$.
    Then, $\liminf_{t\to\infty} z_t\ge R$.
\end{lem}

\begin{proof}
Consider $\delta>0$ and note that, by hypothesis, there exists $t_\delta>0$ such that $\varepsilon_t\ge R-\delta$ for all $t\ge t_\delta$. Hence, for $t\ge t_{\delta}$, we have 
\[z_{t+1}\ge(1-\alpha_t)z_t+\alpha_t(R-\delta).\] Let $\hat{z}_t=z_t-(R-\delta)$ for all $t$. We have, for $t\ge t_{\delta}$, $\hat{z}_{t+1}\ge(1-\alpha_t)\hat{z}_t$. Since $\sum_{t\ge 0}\alpha_t=\infty$ and $\ln{(1-\alpha_t)}\ge-\frac{1}{1-\eta}\alpha_t,\;\alpha_t\in(0,\eta]$, we conclude that $\liminf_{t\to\infty}\prod_{s=t_\delta}^{t-1}(1-\alpha_s)\ge 0$, and, hence, $\liminf_{t\to\infty}\hat{z}_t\ge 0$. Thus, $\liminf_{t\to\infty}z_t\ge R-\delta$, from which the assertion follows by taking $\delta$ to $0$. 
\end{proof}

\subsection{Proof of Proposition \ref{prop: boundedness}}

With Lemma \ref{lem: coro}, following a similar analysis as Lemma 5.1 in \cite{kar2013cal}, we can obtain Proposition \ref{prop: boundedness}.

\subsection{Proof of Proposition \ref{prop: consensus}}
Equation (\ref{eq:z}) can be written as
\begin{multline*}
    \mathbf{z}_{i,u}(k+1)=\Phi_{i,u}(k,0)\mathbf{z}_{i,u}(0)\\
    -\sum_{r=0}^{k-1}a_r\Phi_{i,u}(k,r+1)\bar{\mathbf{d}}_{i,u}(\mathbf{z}_r)
    -a_k\bar{\mathbf{d}}_{i,u}(\mathbf{z}_k).
\end{multline*}
The residual $\mathbf{z}_{i,u}(k+1)-\boldsymbol{1}\mathbf{q}_{i,u}\T(k+1)\mathbf{z}_{i,u}(k+1)$ evolves as 
\begin{eqnarray}
\label{residue}
&&\!\!\!\!\!\!\!\!\!\!\!\!\!\!\!\!\!\!\!\!\mathbf{z}_{i,u}(k+1)-\boldsymbol{1}\mathbf{q}_{i,u}\T(k+1)\mathbf{z}_{i,u}(k+1)\nn\\
&=&\!\!\!\!(\Phi_{i,u}(k,0)-\boldsymbol{1}\mathbf{q}_{i,u}\T(0))\mathbf{z}_{i,u}(0)\nn\\
&&\!\!\!\!-\sum_{r=0}^{k-1}a_r(\Phi_{i,u}(k,r+1)-\boldsymbol{1}\mathbf{q}_{i,u}\T(r+1))\bar{\mathbf{d}}_{i,u}(\mathbf{z}_r)\nn\\
&&\!\!\!\!-a_k(I-\boldsymbol{1}\mathbf{q}_{i,u}\T(k+1))\bar{\mathbf{d}}_{i,u}(\mathbf{z}_k).
\end{eqnarray}
The boundedness of $\bar{\mathbf{d}}_{i,u}(\mathbf{z}_k)$ are implied by the boundedness of $\mathbf{Q}_t^n$ by Proposition \ref{prop: boundedness}. Along with
$\lim_{k\to\infty}\Phi_{i,u}(k,s)=\boldsymbol{1}\mathbf{q}_{i,u}\T(s)$ and $a_k\to 0$ as $k\to\infty$, we conclude that $\limsup_{k\to\infty}\|\mathbf{z}_{i,u}(k)-\boldsymbol{1}\mathbf{q}_{i,u}\T(k)\mathbf{z}_{i,u}(k)\|=0$. Since $\bar{Q}_{i,u}^n(t)$ is a piecewise constant interpolation of $z_{i,u}^n(k)$ and $\bar{Q}_{i,u}^n(t)=\mathbb{E}[Q_{i,u}^n(t)|\mathcal{F}_t]$, the desired assertion follows.

\subsection{Proof of Proposition \ref{prop:QD directed}}

By Proposition \ref{prop: consensus}, agents reach consensus asymptotically, which indicates that \[\limsup_{k\to\infty}\Big\|\boldsymbol{1}\mathbf{q}_{i,u}\T(k)\mathbf{z}_{i,u}(k)-\frac{1}{N}\boldsymbol{1}\T\mathbf{z}_{i,u}(k)\Big\|=0.\] 
We next estimate $\frac{1}{N}\boldsymbol{1}\T\mathbf{z}_{i,u}(k+1)-Q_{i,u}^*$. By (\ref{eq:z}),  
\begin{eqnarray*}
&&\!\!\!\!\!\!\!\!\!\!\!\!\!\!\!\!\frac{1}{N}\boldsymbol{1}\T\mathbf{z}_{i,u}(k+1)\\
&=&\!\!\!\!\frac{1}{N}\boldsymbol{1}\T(I-bL^k_{i,u})\mathbf{z}_{i,u}(k)-a_k\frac{1}{N}\boldsymbol{1}\T\bar{\mathbf{d}}_{i,u}(\mathbf{z}_k)\\
&=&\!\!\!\!(1-a_k)\frac{1}{N}\boldsymbol{1}\T\mathbf{z}_{i,u}(k)+\frac{a_k}{N}\boldsymbol{1}\T\mathcal{G}_{i,u}(\mathbf{z}_k)\\
&&\!\!\!\!-\frac{b}{N}\boldsymbol{1}\T L^k_{i,u}\mathbf{z}_{i,u}(k),
\end{eqnarray*}
from which, we obtain that
\begin{eqnarray}
\label{eq:error}
&&\!\!\!\!\!\!\!\!\!\!\!\!\!\!\!\frac{1}{N}\boldsymbol{1}\T\mathbf{z}_{i,u}(k+1)-Q_{i,u}^*\nn\\
&=&\!\!\!\!\!a_k\Big(\frac{1}{N}\boldsymbol{1}\T\mathcal{G}_{i,u}(\mathbf{z}_k)-\bar{\mathcal{G}}_{i,u}(\mathbf{Q}^*)-\frac{b}{a_k}\frac{1}{N}\boldsymbol{1}\T L^k_{i,u}\mathbf{z}_{i,u}(k)\Big)\nn\\
&&\!\!\!\!\!+(1-a_k)\Big(\frac{1}{N}\boldsymbol{1}\T\mathbf{z}_{i,u}(k)-Q_{i,u}^*\Big).
\end{eqnarray}
In the above equation,
\begin{eqnarray*}
&&\!\!\!\!\!\!\!\!\!\!\!\!\!\!\frac{1}{N}\boldsymbol{1}\T\mathcal{G}_{i,u}(\mathbf{z}_k)-\bar{\mathcal{G}}_{i,u}(\mathbf{Q}^*)\nn\\
&=&\!\!\!\!\gamma p_{ij}^u\sum_{j\in\mathcal{X}}\Big(\frac{1}{N}\sum_{n=1}^N \min_{v\in\mathcal{U}}z_{j,v}^n(k)-\min_{v\in\mathcal{U}}\frac{1}{N}\boldsymbol{1}\T\mathbf{z}_{j,v}(k)\Big)\nn\\
&&\!\!\!\!+\gamma p_{ij}^u\sum_{j\in\mathcal{X}}\frac{1}{N}\sum_{n=1}^N \big(\min_{v\in\mathcal{U}}z_{j,v}^n(k)-\min_{v\in\mathcal{U}}Q_{j,v}^*\big).
\end{eqnarray*}
Since all agents reach consensus asymptotically, we have
\[\lim_{k\to\infty}\Big|\frac{1}{N}\sum_{n=1}^N \min_{v\in\mathcal{U}}z_{j,v}^n(k)-\min_{v\in\mathcal{U}}\frac{1}{N}\boldsymbol{1}\T\mathbf{z}_{j,v}(k)\Big|=0.\]
Thus, from (\ref{eq:error}), we have
\[\limsup_{k\to\infty}\frac{1}{N}\big(\boldsymbol{1}\T\mathcal{G}_{i,u}(\mathbf{z}_k)-\bar{\mathcal{G}}_{i,u}(\mathbf{Q}^*)\big)\le\gamma F(k),\]
\[\liminf_{k\to\infty}\frac{1}{N}\big(\boldsymbol{1}\T\mathcal{G}_{i,u}(\mathbf{z}_k)-\bar{\mathcal{G}}_{i,u}(\mathbf{Q}^*)\big)\ge\gamma f(k),\]
where 
\[  F(k)=\max_{i,u}\Big(\frac{1}{N}\boldsymbol{1}\T\mathbf{z}_{i,u}(k)-Q_{i,u}^*\Big),\]
\[  f(k)=\min_{i,u}\Big(\frac{1}{N}\boldsymbol{1}\T\mathbf{z}_{i,u}(k)-Q_{i,u}^*\Big).\]
Let
\begin{multline}
\label{W-bound}
W_{i,u}(k)=-\frac{b}{a_k}\frac{1}{N}\boldsymbol{1}\T L^k_{i,u}\mathbf{z}_{i,u}(k)\\
=-\frac{b}{N}\boldsymbol{1}\T L^k_{i,u}\frac{1}{a_k}\Big(\mathbf{z}_{i,u}(k)-\boldsymbol{1}\mathbf{q}_{i,u}\T(k)\mathbf{z}_{i,u}(k)\Big),  
\end{multline}
where we have used the fact that $L^k_{i,u}\boldsymbol{1}=0$. From (\ref{residue}), 
\begin{eqnarray*}
&&\!\!\!\!\!\!\!\!\!\!\!\!\!\!\!\!\frac{1}{a_k}\left(\mathbf{z}_{i,u}(k)-\boldsymbol{1}\mathbf{q}_{i,u}\T(k)\mathbf{z}_{i,u}(k)\right)\nn\\
&=&\!\!\!\!\frac{\Phi_{i,u}(k-1,0)-\boldsymbol{1}\mathbf{q}_{i,u}\T(0)}{a_k}\mathbf{z}_{i,u}(0)\nn\\
&&\!\!\!\!+\frac{a_{k-1}}{a_k}(\boldsymbol{1}\mathbf{q}_{i,u}\T(k)-I)\bar{\mathbf{d}}_{i,u}(\mathbf{z}_{k-1})\nn\\
&&\!\!\!\!-\sum_{r=0}^{k-2}\frac{a_r(\Phi_{i,u}(k-1,r+1)-\boldsymbol{1}\mathbf{q}_{i,u}\T(r+1))}{a_k}\bar{\mathbf{d}}(\mathbf{z}_r).
\end{eqnarray*}
It is implied in \cite{vaidya2012matrix} that $\Phi_{i,u}(k,s)$ converges to $\boldsymbol{1}\mathbf{q}_{i,u}\T(s)$ exponentially fast. Since $\sum_{k\ge 0}a_k=\infty$, the convergence speed of $a_k$ is much slower than the exponential convergence speed. Then $\lim_{k\to\infty}\frac{a_{s-1}(\Phi_{i,u}(k-1,s)-\boldsymbol{1}\mathbf{q}_{i,u}\T(s))}{a_k}=0$, $\forall s\in[0,k-1]$. Note that $\lim_{k\to\infty}\frac{a_{k-1}}{a_k}=1$. From (\ref{W-bound}),
\begin{eqnarray*}
&&\!\!\!\!\!\!\!\!\!\!\!\!\!\!\!\!\lim\limits_{k\to\infty}W_{i,u}(k)\\
&=&\!\!\!\!-\frac{b}{N}\boldsymbol{1}\T\lim\limits_{k\to\infty}L^k_{i,u}(\boldsymbol{1}\mathbf{q}_{i,u}\T(k+1)-I)\bar{\mathbf{d}}_{i,u}(\mathbf{z}_{k})\\
&=&\!\!\!\!-\frac{b}{N}\boldsymbol{1}\T\lim_{k\to\infty}L^k_{i,u}\bar{\mathbf{d}}_{i,u}(\mathbf{z}_{k})\\
&=&\!\!\!\!-\frac{b}{N}\boldsymbol{1}\T\lim_{k\to\infty}L^k_{i,u}
(\mathbf{z}_{i,u}(k)-\mathcal{G}_{i,u}(\mathbf{z}_{k}))\\
&=&\!\!\!\!-\frac{b}{N}\boldsymbol{1}\T\lim_{k\to\infty}L^k_{i,u}
\mathcal{G}_{i,u}(\mathbf{z}_{k})\\
&=&\!\!\!\!\lim_{k\to\infty}\frac{b}{N}\sum_{v_l\in\mathcal{N}_n^k}(c_l(i,u)-c_n(i,u)).
\end{eqnarray*}
Note that
\begin{eqnarray*}
&&\!\!\!\!c_l(i,u)-c_n(i,u)\ge(1-\gamma)\min_{j,v}(Q_{j,v}^{l*}-Q_{j,v}^{n*}),\\
&&\!\!\!\!c_l(i,u)-c_n(i,u)\le(1-\gamma)\max_{j,v}(Q_{j,v}^{l*}-Q_{j,v}^{n*}).
\end{eqnarray*}
Let $M_{j,v}=\max\limits_{v_n\in\mathcal{V}}Q_{j,v}^{n*}$ and $m_{j,v}=\min\limits_{v_n\in\mathcal{V}}Q_{j,v}^{n*}$. Then,
\[\limsup\limits_{k\to\infty}W_{i,u}(k)\le(1-\gamma)\max_{j,v}(M_{j,v}-m_{j,v}),\]
\[\liminf\limits_{k\to\infty}W_{i,u}(k)\ge(1-\gamma)\min_{j,v}(m_{j,v}-M_{j,v}).\]
From (\ref{eq:error}), we obtain
\[F(k+1)\le(1-a_k(1-\gamma))F(k)+a_k(1-\gamma)\max_{j,v}(M_{j,v}-m_{j,v}),\]
\[f(k+1)\ge(1-a_k(1-\gamma))f(k)+a_k(1-\gamma)\min_{j,v}(m_{j,v}-M_{j,v}).\]
By Proposition 4.1 in \cite{kar2013cal}, we have
\begin{equation}
\label{F}
    \limsup_{k\to\infty}F(k)\le\max_{j,v}(M_{j,v}-m_{j,v}).
\end{equation}
By Lemma \ref{lem lower bound}, we have
\begin{equation}
\label{f}
    \liminf_{k\to\infty}f(k)\ge\min_{j,v}(m_{j,v}-M_{j,v}).
\end{equation}
Equations (\ref{F}) and (\ref{f}) imply
\[ \limsup_{k\to\infty}\Big|\frac{1}{N}\boldsymbol{1}\T\mathbf{z}_{i,u}(k)-Q_{i,u}^*\Big|\le \max\limits_{v_n,v_l\in\mathcal{V}}\|\mathbf{Q}^{n*}-\mathbf{Q}^{l*}\|_\infty.\]
Since $z_{i,u}^n(k)$, $\forall v_n$ reach consensus  as $k\to\infty$, the above inequality further implies
\[ \limsup_{k\to\infty}|z_{i,u}^n(k)-Q_{i,u}^*|\le\max\limits_{v_n,v_l\in\mathcal{V}}\|\mathbf{Q}^{n*}-\mathbf{Q}^{l*}\|_\infty=R.\]
Note that $\bar{Q}_{i,u}^n(t)$ is a piecewise constant interpolation of $z_{i,u}^n(k)$ and $\bar{Q}_{i,u}^n(t)=\mathbb{E}[Q_{i,u}^n(t)|\mathcal{F}_t]$. We have \[\mathbb{P}\Big(\limsup_{t\to\infty}|Q_{i,u}^n(t)-Q_{i,u}^*|\le R\Big)=1,\] \ie, \[\mathbb{P}\Big(\limsup\limits_{t\to\infty}\|\mathbf{Q}_t^n-\mathbf{Q}^*\|_\infty\le R\Big)=1.\]
From equation (\ref{eq:V}), \[\max_{i}|V_{i}^n(t)-V_{i}^*|\le\max_{i,u}|Q_{i,u}^n(t)-Q_{i,u}^*|\le R,\] 
\ie, \[\mathbb{P}\Big(\limsup\limits_{t\to\infty}\|\mathbf{V}_t^n-\mathbf{V}^*\|_\infty\le R\Big)=1.\]

If $|Q_{i,u}^*-Q_{i,v}^*|\ge 2R,\;u,v\in\mathcal{U}$, the set $(Q_{i,u}^{*}+R,Q_{i,u}^{*}-R)$ and the set $(Q_{i,v}^{*}+R,Q_{i,v}^{*}-R)$ do not overlap. Thus, 
$\argmin_{v}Q_{i,u}^n(t)=\argmin_{v}Q_{i,u}^{*}$ as $t\to\infty$, which indicates that each agent can learn the optimal policy $\pi^*$.

For each $v_n\in\mathcal{V}$, define $F^n(k)=\max_{i,u}(z_{i,u}^n(k)-Q_{i,u}^{n*})$ and $f^n(k)=\min_{i,u}(z_{i,u}^n(k)-Q_{i,u}^{n*})$. Following the similar analysis as $F(k)$ and $f(k)$,
we can prove that 
\[\limsup_{k\to\infty}\left(z_{i,u}^n(k)-Q_{i,u}^{n*}\right)\le\max_{j,v}(\max_{v_l}Q_{j,v}^{l*}-Q_{j,v}^{n*}),\]
\[\liminf_{k\to\infty}\left(z_{i,u}^n(k)-Q_{i,u}^{n*}\right)\ge\min_{j,v}(\min_{v_l}Q_{j,v}^{l*}-Q_{j,v}^{n*}).\]
Since $\max_{v_l}Q_{j,v}^{l*}\le M$ and $\min_{l}Q_{j,v}^{l*}\ge m$, we obtain 
\[\limsup_{k\to\infty}z_{i,u}^n(k)\le M-\max_{j,v}Q_{j,v}^{n*}+Q_{i,u}^{n*}\le M,\]
\[\liminf_{k\to\infty}z_{i,u}^n(k)\ge m-\min_{j,v}Q_{j,v}^{n*}+Q_{i,u}^{n*}\ge m,\]
which indicate
\[\mathbb{P}(\limsup\limits_{t\to\infty}Q_{i,u}^n(t)\le M)=1,\;\mathbb{P}(\liminf\limits_{t\to\infty}Q_{i,u}^n(t)\ge m)=1.\]

\end{document}